\documentclass[conference]{IEEEtran}
\IEEEoverridecommandlockouts
\usepackage{amsmath,amssymb,amsfonts}
\usepackage{algorithmic}
\usepackage{graphicx}
\usepackage{textcomp}
\usepackage{xcolor}
\def\BibTeX{{\rm B\kern-.05em{\sc i\kern-.025em b}\kern-.08em
    T\kern-.1667em\lower.7ex\hbox{E}\kern-.125emX}}

\newtheorem{theorem}{Theorem}
\newtheorem{lemma}{Lemma}
\newtheorem{remark}{Remark}
\newtheorem{assumption}{Assumption}
\newtheorem{definition}{Definition}
\newtheorem{proof}{Proof}
\newcommand{\norm}[1]{\left\lVert#1\right\rVert}
\newcommand\scalemath[2]{\scalebox{#1}{\mbox{\ensuremath{\displaystyle #2}}}}
\usepackage{subfigure}
\usepackage{pdfpages}
    
\begin{document}

\title{Distributed Robust Continuous-Time Optimization Algorithms for Time-Varying Constrained Cost\\
}

\author{\IEEEauthorblockN{ Zeinab Ebrahimi}
\IEEEauthorblockA{\textit{School of Mechanical and Manufacturing Engineering} \\
\textit{University of New South Wales}\\
Sydney, Australia \\
z.ebrahimi@unsw.edu.au}
\and
\IEEEauthorblockN{ Mohammad Deghat}
\IEEEauthorblockA{\textit{School of Mechanical and Manufacturing Engineering} \\
\textit{University of New South Wales}\\
Sydney, Australia \\
m.deghat@unsw.edu.au}
}

\maketitle

\begin{abstract}
This paper presents a distributed continuous-time optimization framework aimed at overcoming the challenges posed by time-varying cost functions and constraints in multi-agent systems, particularly those subject to disturbances. By incorporating tools such as log-barrier penalty functions to address inequality constraints, an integral sliding mode control for disturbance mitigation is proposed. The algorithm ensures asymptotic tracking of the optimal solution, achieving a tracking error of zero. The convergence of the introduced algorithms is demonstrated through Lyapunov analysis and nonsmooth techniques. Furthermore, the framework's effectiveness is validated through numerical simulations considering two scenarios for the communication networks.
\end{abstract}


\section{Introduction}
Over the last decades, there has been a significant increase in research focused on distributed optimization within multi-agent systems, driven by its wide range of potential applications, including smart grids \cite{cherukuri2016initialization}, machine learning \cite{yi2022primal}, and multi-robot systems \cite{li2017fixed}. Our research lies in a particular category of distributed convex optimization problems. Within the distributed optimization framework, each agent is allocated a local cost function. The joint aim is to achieve optimization by minimizing a global cost function, which represents the aggregated objectives of all agents involved.

Extensive research has been carried out on distributed continuous-time optimization algorithms \cite{lin2016distributed,feng2019finite}, motivated by the widespread application of continuous-time models in real-world systems, such as the coordination optimization problems in multi-agent systems. These algorithms have been adapted for a variety of scenarios, including attaining finite-time convergence \cite{feng2019finite}, first-order, high-order, and fractional-order dynamics \cite{huang2019distributed}. 

The aforementioned studies have predominantly focused on time-invariant optimization scenarios with fixed objectives and constraints. 
However, the reality of applications such as resource allocation in changing environments \cite{wangBo2020distributed}, tracking moving targets \cite{lee2015multirobot}, and online optimization \cite{jadbabaie2015online} demands algorithms capable of adapting to time-varying cost functions. In these circumstances, the optimal solution varies over time, creating a trajectory instead of a static point. 
A pivotal development in this area has been the work of Rahili and Ren \cite{rahili2016distributed}, who proposed innovative results for tackling unconstrained time-varying optimization in systems with single and double-integrator dynamics over undirected networks. Their methodologies, including average gradient tracking and state estimator-based algorithms, mark a significant step forward. 

The literature on time-varying distributed optimization with constraints remains notably limited. Researchers such as He \textit{et al}. \cite{he2021continuous} introduced a predictive correction method with the limitation of requiring identical Hessian matrices for each agent's cost function, and Wang \textit{et al}. \cite{wang2020distributed} overlooked the nonlinear inequality constraints in the resource allocation problem. Meanwhile, Sun \textit{et al}. \cite{sun2020distributed} explored the adaptation of penalty approaches to avoid the need for uniform Hessian matrices. These studies highlight a gap in the literature regarding the complexities introduced by nonuniform local inequality constraints. 

In the domain of distributed optimization, the robustness of algorithms against disturbances is a critical issue, as their absence can lead to unreliable system performance or even instability. With growing interest in frameworks that can manage both disturbances and time-varying costs, traditional methods like sliding mode and internal model control often fall short, focusing mainly on static cost functions. Recently, Zhang \textit{et al}. \cite{zhang2023continuous} proposed a robust distributed optimization framework for multi-agent systems with time-varying costs affected by disturbances. However, this approach did not take into account time-varying constraints. 
This paper studies the distributed continuous-time optimization problem in the context of time-varying cost functions incorporating time-varying nonuniform local inequality constraints subject to disturbances.

In this work, we synthesize insights from existing research on distributed continuous-time optimization in multi-agent systems, particularly those dealing with time-varying cost functions and the challenge of disturbances. Unlike the focused approaches of Zhang \textit{et al}. \cite{zhang2023continuous}, who proposed a robust optimization framework for systems with time-varying costs but without considering variable constraints, and Sun \textit{et al}. \cite{sun2022distributed2023}, who addressed distributed optimization with an emphasis on overcoming the limitations of identical Hessian matrices without exploring disturbance management, our contributions bridge these gaps by presenting a unified framework. Our methodology integrates a dual strategy of log-barrier penalty functions for dynamic constraint handling, paired with an integral sliding mode control designed for disturbance rejection. This integrated approach enables robust optimization in dynamic conditions. It not only suppresses disturbances but also flexibly manages varying local constraints. Theoretical foundations for asymptotic convergence are provided through nonsmooth analysis and Lyapunov theory. Theoretical analysis ensures the reliability of the proposed framework, which is substantiated by numerical simulations. We present two distinct simulation scenarios that demonstrate the applicability of our theorem across both connected and less-connected network topologies, thereby validating our theoretical framework.

\section{Notations and preliminaries}

\subsection{Notations}
The following notations are adopted throughout this paper.\\
Let $\mathbb{R}$ represent the set of all real numbers, and $\mathbb{R}_{+}$ denote the subset of positive real numbers. The symbols $\mathbb{R}^n$ and $\mathbb{R}^{n \times m}$ are used to indicate \(n\)-dimensional real vectors and \(n \times m\) real matrices, respectively. The cardinality of a set $S$ is represented by $|S|$. The vectors $\textbf{1}_n$ and $\textbf{0}_n$ denote the $n$-dimensional vector of ones and zeros, respectively. The identity matrix of size $n \times n$ is denoted by $I_n$. 
For any vector \(h = [h_1, \ldots, h_n]^\top \in \mathbb{R}^n\), the term \(\text{diag}(h) \in \mathbb{R}^{n \times n}\) refers to the diagonal matrix where \(h\)'s elements form the diagonal. We define \(\text{sig}(h)^\alpha= \text{sign}(h) |h|^\alpha\), for any \(h \in \mathbb{R}^n\) where \(\alpha > 0\). It is noteworthy that \(d|h|^{\alpha+1}/dh = (\alpha + 1) \text{sig}(h)^\alpha\) and \(h\text{sig}(h)^\alpha = |h|^{\alpha+1}\).
We define \(\text{sign}(h) = [\text{sign}(h_1), \ldots, \text{sign}(h_n)]^\top \in \mathbb{R}^n\) and \(\text{sig}(h)^\alpha = [\text{sig}(h_1)^\alpha, \ldots, \text{sig}(h_n)^\alpha]^\top \in \mathbb{R}^n\). The norms \(\|h\|_1\), \(\|h\|_2\), and \(\|h\|_\infty\) denote the 1-norm, 2-norm, and \(\infty\)-norm of \(h \in \mathbb{R}^n\), respectively.
The Kronecker product is symbolized by \(\otimes\), \(\Bar{co}\) denotes the convex closure, and $B(h, \delta)$ represents the open ball of radius $\delta$ centered at $h$. The symbol $\prec$ denotes an element-wise comparison operator. The gradient and Hessian of a function \(f(h, t)\) with respect to the vector \(h\) are denoted as \(\nabla f(h, t)\) and \(\nabla^2 f(h, t)\), respectively. The notation \(\partial f(h, t)/\partial t\) is used to represent the partial derivative of \(f(h, t)\) with respect to \(t\), while \(\dot{f}(h, t)\) denotes the time derivative of \(f(h, t)\), defined explicitly as \(\dot{f}(h, t) = \nabla f(h, t) \dot{h} + \partial f(h, t)/\partial t\).

\subsection{Graph Theory}

Consider a graph $\mathcal{G}=(\mathcal{V}, \mathcal{E}, \mathcal{A})$, where $\mathcal{V}$ denotes the set of nodes indexed by $\mathcal{V} = \{1, \ldots, N\}$, $\mathcal{E} \subseteq \mathcal{V} \times \mathcal{V}$ represents the set of edges and $\mathcal{A}=[a_{ij}] \in \mathbb{R}^{N \times N}$ is the adjacency matrix, where $a_{ij} = 1$ if there exists an edge from node $i$ to node $j$, and $a_{ij} = 0$ otherwise. The graph's Laplacian matrix is denoted by $\mathcal{L}=[l_{ij}] \in \mathbb{R}^{N \times N}$ and is defined as $l_{ij} = -a_{ij}$, $\forall i \neq j$, and $l_{ii} = \sum_{j=1}^{N} a_{ij}$. Let $\mathcal{N}_i = \{j \in \mathcal{V} | (j, i) \in \mathcal{E}\}$ represent the set of neighboring nodes to node $i$. 
A graph is considered connected if a sequence of edges exists that allows one to move from any one node to any other. For connected undirected graphs, the Laplacian matrix $\mathcal{L}$ is positive semidefinite. The incidence matrix, denoted by $\mathcal{D} = [d_{ij}] \in \mathbb{R}^{N \times |\mathcal{E}|}$ is defined such that $d_{ik} = -1$ if the $k^{th}$ edge starts from node $i$, $d_{ik} = 1$ if it ends at node $i$, and $d_{ik} = 0$ otherwise. Notably, for an undirected graph, the conditions $\mathcal{L} \mathbf{1}_N = \mathbf{0}_N$, $\mathcal{L}^\top = \mathcal{L}$, and $\mathcal{L} = \mathcal{D}\mathcal{D}^\top$ hold \cite{edmonds2022undirected}. 

\subsection{Definitions and Lemmas}

This subsection introduces crucial definitions and lemmas that will be utilized in our primary results.

\begin{definition}[Filippov Solution] \cite[Definition 2.1]{shevitz1994lyapunov} \label{def 1}
Given the vector differential equation
\begin{equation} \label{flippov}
    \dot{x} = f(x, t)
\end{equation}
with $f: \mathbb{R}^n \times \mathbb{R} \to \mathbb{R}^n$ being measurable and bounded locally. A vector function $x(\cdot)$ is termed a Filippov solution to (\ref{flippov}) over interval $[t_0, t_1]$, if $x(\cdot)$ is absolutely continuous on $[t_0, t_1]$ and satisfies $\dot{x}(t) \in K[f](x, t)$ for almost all $t \in [t_0, t_1]$. Here, $K[f](x, t)$, the Filippov set-valued map for $f(x, t)$, is defined as
\[K[f](x, t) := \bigcap_{\delta>0}\bigcap_{\mu(N)=0} \Bar{co} \left(f(B(x, \delta) - N, t)\right),\]
with the intersection taken over all null Lebesgue measure sets $N$.
\end{definition}
\begin{definition} [Clarke's Generalized Gradient] \cite[Definition 2.2]{shevitz1994lyapunov} \label{def 2}
    For a locally Lipschitz continuous function $V: \mathbb{R}^n \to \mathbb{R}$, the generalized gradient at a point $(x, t)$, denoted by $\partial V(x, t)$, is
\[\partial V(x, t) = \Bar{co}\left\{\lim \nabla V(x, t) \,|\, (x_i, t_i) \to (x, t), (x_i, t_i) \notin \Omega_V\right\}\]
where $\Omega_V$ denotes the Lebesgue measure zero set where $V$'s gradient is undefined.
\end{definition}
\begin{definition}[Chain Rule] \cite[Theorem 2.2]{shevitz1994lyapunov} \label{def 3}
     If $x(\cdot)$ is a Filippov solution to $\dot{x} = f(x, t)$ and $V(x): \mathbb{R}^n \to \mathbb{R}$ is locally Lipschitz continuous function, then for almost all $t$,
$\frac{d}{dt}V(x(t)) \in \dot{\tilde{V}}$,
where $\dot{\tilde{V}}$ is the set-valued Lie derivative, defined by $\dot{\tilde{V}} := \bigcap_{\xi \in \partial V} \xi^\top K[f]$.
\end{definition} 

The following lemmas are used throughout the paper.
\begin{lemma} [\cite{polyakov2011nonlinear}] \label{lemma 1}
    Consider the nonlinear dynamical system $\dot{x}(t) = f(x, t)$, with $x \in \mathbb{R}^n$ and $f: \mathbb{R}^n \times \mathbb{R}_+ \to \mathbb{R}^n$. If a Lyapunov function $V(x)$ exists, satisfying $\dot{V}(x) \leq -\alpha V^p(x) - \beta V^q(x)$, where $\alpha > 0$, $\beta > 0$, $0 < p < 1$, and $q > 1$, then the origin is fixed time stable within settling time $T_f \leq \frac{1}{\alpha(1-p)} + \frac{1}{\beta(q-1)}$.
\end{lemma}

\begin{lemma}[\cite{andrieu2008homogeneous}] \label{lemma 2}
    Consider the system $\dot{x}(t) = f(x)$ with initial condition $x(0) = x_0$ and $x \in \mathbb{R}^n$, where $f: \mathbb{R}^n \to \mathbb{R}^n$ and the origin is an equilibrium point for the system. Assume $f(x)$ exhibits a homogeneous vector field in the bi-limit with associated triples $(r_0, k_0, f_0)$ and $(r_{\infty}, k_{\infty}, f_{\infty})$. If the origins of systems $\dot{x} = f(x)$, $\dot{x}_0 = f_0(x)$, and $\dot{x}_{\infty} = f_{\infty}(x)$ are globally asymptotically stable and the condition $k_{\infty} > 0 > k_0$ holds, then the origin achieves fixed-time stability.
\end{lemma}

\begin{lemma}[\cite{zuo2014new}] \label{lemma 3}
    Given any non-negative $q_1, q_2, \ldots, q_n \geq 0$ and $0 < p \leq 1$, it follows that $\sum_{i=1}^n q_i^p \geq (\sum_{i=1}^n q_i)^p$. Additionally, for $p > 1$, $\sum_{i=1}^n q_i^p \geq n^{1-p}(\sum_{i=1}^n q_i)^p$.
\end{lemma}

\section{Problem Formulation and Main Results}

Consider a multi-agent system comprising $N$ agents. Each agent is represented as a node in the undirected graph $\mathcal{G}$ and is only limited to communications with its neighbors within the network. Suppose that each agent is modeled as the subsequent continuous-time first-order 
\begin{equation} \label{dynamic}
           \dot{x}_i(t) = u_i(t) + d_i(t), \quad i \in \mathcal{V} 
\end{equation}
where $x_i(t) \in \mathbb{R}^n$ and $u_i(t) \in \mathbb{R}^n$ are the state and control input of the $i$th agent, and $d_i(t) \in \mathbb{R}^n$ is an unknown disturbance impacting the system. Each agent \(i\) is allocated with a differentiable local time-varying cost function \(f_i(x, t): \mathbb{R}^n \times \mathbb{R}_+ \to \mathbb{R}\), which is private to that agent. The overall cost function \(F(x): \mathbb{R}^n \times \mathbb{R}_+ \to \mathbb{R}\) for all agents is defined as \(F(x) = \sum_{i=1}^N f_i(x, t)\). The objective of this paper is to formulate a distributed control law $u_i(t)$ for the first-order system (\ref{dynamic}) by using only local information and interactions with neighbor agents, ensuring that all agents converge to an optimal state $x^*(t) \in \mathbb{R}^n$. The optimization problem is defined as
\begin{align}  \label{opt inequality bounded}\nonumber
    & \min_{x_i} \sum_{i=1}^{N} f_i(x_i(t), t) \\ 
    &\text{s.t.} \quad g_i(x_i(t), t) \leq 0, \quad x_i(t) = x_j(t), \quad \forall i, j \in \mathcal{V}
\end{align}
where $g_i(x_i, t)$$:\mathbb{R}^n \times \mathbb{R}_+ \rightarrow \mathbb{R}^{q_i}$ are local inequality constraint functions and $q_i$ represents the number of local inequality constraints for agent $i$. Here, the goal is that each state $x_i$, $\forall i \in \mathcal{V}$, converges to the optimal solution $x^*(t)\in \mathbb{R}^n$, i.e.,
\begin{equation}
    \lim_{t \to \infty} \big( x_i(t) - x^*(t) \big) = 0.
\end{equation}
For ease of notation, we will eliminate the time index $t$ from the variables $x_i(t)$, $d_i(t)$ and $u_i(t)$ in the majority of the subsequent sections, retaining it only in specific instances where deemed necessary. To address the above distributed time-varying optimization problems with disturbances, we introduce the following assumptions.

\begin{assumption}[Graph connectivity] \label{ass graph}
The graph $\mathcal{G}$ is fixed, undirected, and connected.
\end{assumption}

\begin{assumption}[Convexity] \label{ass convexity}
    All the objective functions $f_i(x_i, t)$ are twice continuously differentiable with respect to $x_i$ and continuously differentiable concerning $t$. Furthermore, all the objective functions $f_i(x_i, t)$ are uniformly strongly convex in $x_i$, for all $t \geq 0$. Moreover, all the inequality constraint functions $g_i(x_i, t)$ are twice continuously differentiable to $x_i$, continuously differentiable concerning $t$, and uniformly convex in $x_i$, for all $t \geq 0$. Finally, an optimal solution for the optimization problem (\ref{opt inequality bounded}) exists and is unique. 
\end{assumption}
\begin{assumption}[Slater’s condition] \label{ass slater}
    For all $t \geq 0$, there exists at least one $x$ such that $g_i(x, t) \prec 0_{q_i}$ for all $i \in \mathcal{V}$. 
\end{assumption}

\begin{assumption}[Bounded disturbance]\label{ass bounded disturbance}
    The disturbance $d_i(t)$ is bounded. That is, $\|d_i(t)\|_{\infty} < D_0$, $i \in \mathcal{V},$ where the upper bound $D_0$ is assumed to be known.
\end{assumption}

\begin{remark}
    Assumption \ref{ass convexity}, which relates to the uniform strong convexity of the objective function, ensures that the optimal path $x^*(t)$ is unique for all $t \geq 0$. It should be noted that Assumption \ref{ass convexity} is commonly satisfied in many practical applications, such as coordinated path planning and formation control. Similar assumptions are used in recent research on distributed time-varying optimization \cite{sun2022distributed2023,fazlyab2017prediction, ding2023distributed}. Assumption \ref{ass bounded disturbance} necessitates merely the awareness of an upper bound for disturbances, a condition that is both minimal and frequently applied in prior studies. Various disturbance forms meet this criterion, including constant, sinusoidal, and harmonic disturbances, as discussed in \cite{feng2019finite}.
\end{remark}

\begin{remark}
    The framework for distributed, time-varying optimization with constraints in (\ref{opt inequality bounded}) is widely applicable in distributed cooperative control scenarios. This includes tasks such as navigation for multiple robots \cite{lee2015multirobot, verscheure2009time}, and managing resources within power networks \cite{cherukuri2016initialization}. 
\end{remark}

Define penalized objective function of the $i$th agent $\tilde{L}_i(x_i, t)$ for the optimization problem (\ref{opt inequality bounded}) as:
\begin{equation} \label{penalty}
    \tilde{L}_i(x_i, t) = f_i(x_i, t) - \frac{1}{\rho_i(t)} \sum_{j=1}^{q_{i}} \log \big(\sigma_i(t) - g_{ij}(x_i, t)\big)
\end{equation}
where $g_{ij}(x_i, t): \mathbb{R}^n \times \mathbb{R}_+ \rightarrow \mathbb{R}$ denotes the $j$th component of function $g_i(x_i, t)$, $\rho_i(t) \in \mathbb{R}_+$ is a time-varying barrier parameter, and $\sigma_i(t) \in \mathbb{R}_+$ is a time-varying slack function satisfying
\begin{equation}
   \rho_i(t) = a_{i1}e^{a_{i2} t}, \quad \sigma_i(t) = a_{i3}e^{-a_{i4} t}, \hspace{0.3cm} a_{i1}, a_{i2}, a_{i3}, a_{i4} \in \mathbb{R}_+. 
\end{equation}
\begin{remark}
    It should be noted that the barrier parameter $\rho_i(t)$ and the slack variable $\sigma_i(t)$ should be chosen such that the optimal solution $x^*(t)$ in (\ref{opt inequality bounded}) is attained and the tracking error approximation tends to zero. Therefore, $\rho_i(t)$ must be positive, monotonically increasing, asymptotically go to infinity as $t \rightarrow \infty$, and be bounded within a finite time frame, while $\sigma_i(t)$ must be nonnegative and approach zero as $t \rightarrow \infty$ as proved in \cite[Lemma 1]{fazlyab2017prediction}.
\end{remark}
Inspired by the integral mode approach \cite{zhang2023continuous} and the distributed control algorithm
for the time-varying constrained optimization problem in \cite{sun2022distributed2023}, the distributed controllers for multi-agent system (\ref{dynamic}) are designed as
\begin{subequations} \label{controller} \begin{align} \label{ui}
            & u_i(t) = u_{i1}(t) + u_{i2}(t)\\ \label{ui1}
            & u_{i1}(t) = - \beta \left( \nabla^2 \tilde{L}_i(x_i, t) \right)^{-1} \sum_{j \in \mathcal{N}_i} \text{sign}(x_i - x_j) \\ \nonumber
            &  \hspace{0.4cm} - \left( \nabla^2 \tilde{L}_i(x_i, t) \right)^{-1} \left( \nabla \tilde{L}_i(x_i, t) + \frac{\partial}{\partial t} \nabla \tilde{L}_i(x_i, t) \right) \\ \label{ui2}
            & u_{i2}(t) = - k_0 \text{sign}(s_i) - k_1 \text{sig}(s_i)^{\rho_1} - k_2 \text{sig}(s_i)^{\rho_2} \\ \label{siii} 
            & s_i(t) = x_i(t) - \int_{0}^{t} u_{i1}(\tau) d\tau
            \end{align}
\end{subequations}
where $k_0 > D_0$, $k_1, k_2 > 0$, $0 < \rho_1 < 1$, $ \rho_2 > 1$, and $\beta >0$ is a fixed control gain. $s_i(t)$ is the designed distributed sliding manifold. Let $\psi_i$ be defined as follows:
\begin{equation}
    \psi_i=\left( \nabla^2 \tilde{L}_i(x_i, t) \right)^{-1} \left( \nabla \tilde{L}_i(x_i, t) + \frac{\partial}{\partial t} \nabla \tilde{L}_i(x_i, t) \right).
\end{equation}

Note that the domain of the penalized objective function $\tilde{L}_i(x_i, t)$ is $D_i(t) = \{x_i \in \mathbb{R}^n | g_i(x_i, t) \prec \sigma_i(t)\mathbf{1}_{q_i} \}$. To ensure that the distributed controller (\ref{ui})--(\ref{siii}) functions correctly, the initial states $x_i(0)$ must satisfy the following condition:
\begin{equation} \label{initial}
 g_{ij} \left(x_i(0), 0 \right) < \sigma_i(0)   
\end{equation}

This requirement guarantees that the initial states of each agent are within the feasible region defined by the slack function $\sigma_i(0)$, facilitating the successful application of the controller.

\begin{remark}
    A comparable problem was addressed in \cite{zhang2023continuous} where the nonlinear inequality constraints were not taken into account. Consequently, the control strategies proposed in \cite{zhang2023continuous} can not be directly implemented. The controller algorithm proposed in (\ref{controller}) consists of three components: the first term, $- \beta \left( \nabla^2 \tilde{L}_i(x_i, t) \right)^{-1} \sum_{j \in \mathcal{N}_i} \text{sign}(x_i - x_j)$, aims to tackle the consensus problem by leading all agents toward achieving consensus on states $\left(\lim_{t \rightarrow \infty} \norm{x_i(t) - \frac{1}{N} \sum_{j=1}^N x_j(t)}_2=0 \right)$. The second term $(\psi_i(t))$ serves as an optimizer to minimize the penalized objective function $\tilde{L}_i(x_i, t)$ provided in (\ref{penalty}). This step includes the use of log-barrier penalty functions to integrate inequality constraints into the penalized objective framework. The final term, described in (\ref{ui2}), plays a crucial role in suppressing disturbance, thus ensuring the system preserves the nominal optimization algorithm's performance during the sliding phase, even with disturbances. Notably, this paper's controller algorithm (\ref{ui1}) eliminates the restriction that the Hessians of all the local objective functions must be identical. 
\end{remark}
\begin{assumption}[Bounds on Objective Functions] \label{ass bound obj}
If all local states $x_i$ are bounded, then for all $i \in \mathcal{V}$ and $t \geq 0$, there exists $\bar{\alpha}$ such that
$\sup_{t \in [0,\infty)} \left\| \frac{\partial}{\partial t} \nabla f_i(x_i, t) \right\|_2 \leq \bar{\alpha}$.

\end{assumption}

\begin{assumption}[Bounds on Inequality Constraint Functions] \label{ass bound ineq cons} If all local states $x_i$ are bounded, there exist constants $\bar{\beta}$ and $\bar{\gamma}$ ensuring that $\sup_{t \in [0,\infty)} \left\| \frac{\partial}{\partial t} \nabla g_{ij}(x_i, t) \right\|_2 \leq \bar{\beta}$, and $\sup_{t \in [0,\infty)} \left\| \frac{\partial}{\partial t} g_{ij}(x_i, t) \right\|_2\leq \bar{\gamma}$, for all $i \in \mathcal{V}$, and $j=1, \cdots, q_i$.
    
\end{assumption}

\begin{remark}
    The presence of the piecewise-differentiable sign function within algorithm (\ref{controller}) necessitates considering solutions in terms of Filippov's framework. Given the sign function's measurability and local boundedness, Filippov solutions for our system's dynamics are guaranteed to exist \cite{filippov2013differential,cortes2008discontinuous}.
\end{remark}

\begin{theorem} \label{theorem 1}
    Suppose that Assumptions 1-6 and the initial condition (\ref{initial}) hold for the system (\ref{dynamic}) under the controller (\ref{controller}). Let $\bar{\psi}$ be a constant satisfying $\|\psi_i(t) \|_2 \leq \bar{\psi}$ $\forall i \in \mathcal{V}$. Then for a constant $\beta$ satisfying 
    \begin{equation} \label{beta}
        \beta \geq \frac{2 \bar{\psi}  n^2 |\mathcal{E}|}{\min_{i \in \mathcal{V}} \left\{ \bar{\lambda}_{\min} \left((\nabla^2 \tilde{L}_i)^{-1}\right)\right\}} + \epsilon
    \end{equation}
    where $\epsilon> 0$ is a constant, all the states $x_i$ will achieve consensus in finite time, i.e., there exists a time $T_2$ such that $\|x_i(t) - x_j(t)\|_2 = 0$, for all $i, j \in \mathcal{V}$ and for all $t > T_2$, and achieve $x_i(t) \rightarrow x^*(t)$ for $i \in \mathcal{V}$. 
\end{theorem}
\begin{proof}
    The proof is divided into two parts. \textbf{Part 1:} Demonstrating that $u_{i2}(t)$ effectively counters the disturbance $d_i(t)$ within a predetermined timeframe; \textbf{Part 2:} Verifying that $u_{i1}(t)$ derives the state $x_i(t)$ towards the desired trajectory $x^*(t)$. The structure of our proof follows the framework of the proof outlined in \cite{sun2022distributed2023, zhang2023continuous}.

    \textbf{Part 1:} Taking the derivative of $s_i(t)$ results in 
     \begin{equation} \label{sdot}
         \dot{s}_i(t) = \dot{x}_i(t) - u_{i1}(t).
     \end{equation} 
     By substituting (\ref{dynamic}), (\ref{ui}), and (\ref{ui2}) into (\ref{sdot}), we acquire
     \begin{equation} \label{sdot substitute}
         \dot{s}_i = -k_0\text{sign}(s_i) - k_1 \text{sig}(s_i)^{\rho_1} - k_2\text{sig}(s_i)^{\rho_2} + d_i.
     \end{equation}
     We introduce a Lyapunov candidate $V_{S1}(t) = \frac{1}{2} \sum_{i=1}^{N} \sum_{k=1}^{n} s_{ik}^2$. The time derivative of $V_{S1}(t)$ according to (\ref{sdot substitute}) is determined by
     \begin{equation} \label{compact dynamic}
         \begin{split}
             \dot{V}_{S1} =  \sum_{i=1}^{N} \sum_{k=1}^{n} \left(-k_0|s_{ik}| - k_1 |s_{ik}|^{\rho_1 + 1} - k_2  |s_{ik}|^{\rho_2 + 1} +  s_{ik}d_{ik} \right)
         \end{split}
     \end{equation}
    Given Assumption \ref{ass bounded disturbance}, it can be deduced that
    \begin{equation}
    \begin{split}
         \dot{V}_{S1} \leq - \sum_{i=1}^{N} \sum_{k=1}^{n} (k_0 - \|d_i\|_{\infty}) |s_{ik}| - k_1  (s_{ik}^2)^ {\frac{\rho_1 + 1 }{2}}
         - k_2  (s_{ik}^2)^{\frac{\rho_2 + 1}{2}}
    \end{split}
    \end{equation}
     Using the condition $k_0> D_0$, it follows that
     \begin{align} \nonumber
         \dot{V}_{S1} \leq &-2 ^{\frac{\rho_1 + 1}{2}} k_1 \sum_{i=1}^{N} \sum_{k=1}^{n} \left( \frac{1}{2} s_{ik}^2 \right)^{\frac{\rho_1 + 1}{2}} \\
         &- 2 ^{\frac{\rho_2 + 1}{2}} k_2 \sum_{i=1}^{N} \sum_{k=1}^{n} \left( \frac{1}{2} s_{ik}^2 \right)^{\frac{\rho_2 + 1}{2}} 
     \end{align}
     As indicated by Lemma \ref{lemma 3} and given that $0 < \rho_1 < 1$ and $ \rho_2 > 1$, we obtain
     \begin{equation}
         \dot{V}_{S1} \leq -2  ^{\frac{\rho_1 + 1}{2}} k_1 V_{S1}^{\frac{\rho_1 + 1}{2}} - 2 ^{\frac{\rho_2 + 1}{2}} k_2 (Nn)^{\frac{1 - \rho_2}{2}} V_{S1}^{\frac{\rho_2 + 1}{2}}
     \end{equation}
     Then by using Lemma \ref{lemma 1}, this results in a fixed time for reaching, limited to $T_d \leq \frac{1}{2^{\frac{\rho_1-1}{2}} k_1 (1-\rho_1 )} + \frac{1}{2^{\frac{\rho_2 -1}{2}} k_2(N n)^{\frac{1-\rho_2}{2}} (\rho_2-1)}$. Therefore, $u_{i2}(t)$ effectively counters the disturbance $d_i(t)$ in fixed-time.

     \textbf{Part 2:} We find $s_i$ and its derivative $\dot{s}_i$ both equal to zero for all $t \geq T_d$. From (\ref{dynamic}) and (\ref{ui1}), we derive the system's behavior as:
     \begin{equation} \label{system without ui2} \begin{split}
         \dot{x}_i(t)=& - \beta \left( \nabla^2 \tilde{L}_i(x_i, t) \right)^{-1} \sum_{j \in \mathcal{N}_i} \text{sign}(x_i - x_j) \\
         &- \left( \nabla^2 \tilde{L}_i(x_i, t) \right)^{-1} \left( \nabla \tilde{L}_i(x_i, t) + \frac{\partial}{\partial t} \nabla \tilde{L}_i(x_i, t) \right)
              \end{split}
     \end{equation}

     Define the inverse of the Hessian matrix for the penalized objective function as a diagonal matrix combining the inverses of individual Hessians,
$\scalemath{0.9}{\left( \nabla^2 \tilde{L}(x, t) \right)^{-1} = \text{diag} \left\{ \left( \nabla^2 \tilde{L}_1(x_1, t) \right)^{-1}, \ldots, \left( \nabla^2 \tilde{L}_n(x_n, t) \right)^{-1} \right\}}$,
where $x$ is composed of the stacked states of all agents,
$x = \begin{bmatrix} x_1^\top & \ldots & x_n^\top \end{bmatrix}^\top$,
and $\Psi$ represents the collection of auxiliary variables for all agents,
$\Psi = \begin{bmatrix} \psi_1^\top & \ldots & \psi_n^\top \end{bmatrix}^\top$.

According to \cite[Lemma 2]{fazlyab2017prediction}, since the objective functions $f_i(x_i, t)$ are strongly convex and $\sigma_i(t)$ is strictly positive, it follows that $\nabla^2 \tilde{L}_i(x_i, t)$ is $m$-strongly convex for $x \in D_i(t)$. Therefore, $\left( \nabla^2 \tilde{L}_i(x_i, t) \right)^{-1}$ exists and is bounded. The dynamics of $\dot{\nabla}^2 \tilde{L}_i$ can be written as $\dot{\nabla}^2 \tilde{L}_i(x_i, t)=\nabla^2 \tilde{L}_i(x_i, t) \dot{x}_i+\frac{\partial}{\partial t} \nabla \tilde{L}_i(x_i, t)$. Substituting $\dot{x}_i$ from (\ref{system without ui2}) into this result gives the closed-loop dynamics $\dot{\nabla}^2 \tilde{L}_i(x_i, t)=  - \beta \sum_{j \in \mathcal{N}_i} \text{sign}(x_i - x_j) - \nabla \tilde{L}_i(x_i, t)$, which implies that
each $\nabla \tilde{L}_i(x_i, t)$ is bounded for all $t \geq 0$ by employing the input-to-state stability analysis \cite{khalil}. Based on Assumptions \ref{ass bound obj} and \ref{ass bound ineq cons}, it is clear that $\frac{\partial}{\partial t} \nabla \tilde{L}_i(x_i, t)$ is bounded for all $t \geq 0$. Hence, we can conclude that $\psi_i(t)$ remains bounded for every $i \in \mathcal{V}$ and for all $t \geq 0$, given that $\left(\nabla^2 \tilde{L}_i(x_i, t)\right)^{-1}$, $\nabla \tilde{L}_i(x_i, t)$, and $\frac{\partial}{\partial t} \nabla \tilde{L}_i(x_i, t)$ are all bounded.

    Consider the following Lyapunov candidate
\begin{equation}
         V_{S2}(t) = \norm{ (\mathcal{D}^\top \otimes I_n) x}_1.
\end{equation}

The system dynamic (\ref{system without ui2}) is compactly expressed as:
\begin{equation}
    \dot{x} = -\beta \left( \nabla^2 \tilde{L}(x, t) \right)^{-1} (\mathcal{D} \otimes I_n)\text{sign} \left( (\mathcal{D}^\top \otimes I_n) x \right) + \Psi.
\end{equation}

    Given the nature of $V_{S2}(t)$, where it is locally Lipschitz continuous but nonsmooth at some points and on account of Definition \ref{def 2}, the generalized gradient of $V_{S2}(t)$ is obtained as 
\begin{equation}
         \partial V_{S2}(t) = (\mathcal{D}^\top \otimes I_n)^\top \left\{ \text{SGN} \left( (\mathcal{D}^\top \otimes I_n) x \right) \right\}
     \end{equation}
     where \text{SGN(.)} is defined in \cite{cortes2008discontinuous}.
     Following Definition \ref{def 3}, the set-valued Lie derivative of $V_{S2}(t)$ is represented as:
\begin{equation} \label{vs2 dot}
    \dot{\tilde{V}}_{S2}(t) = \bigcap_{\xi \in \text{SGN}[(\mathcal{D}^\top \otimes I_n) x]} \xi^\top (\mathcal{D}^\top \otimes I_n) K[f]
\end{equation}
Here, $K[f] = \Psi - \beta [\nabla^2 \tilde{L}(x, t)]^{-1}(\mathcal{D} \otimes I_n)\text{SGN}[(\mathcal{D}^\top \otimes I_n) x]$ denotes the set-valued Filippov map of the system (\ref{system without ui2}). Due to the presence of an intersection operation on the right side of (\ref{vs2 dot}), it indicates that while $\dot{\tilde{V}}_{S2}(t)$ is not empty and there exists $\xi \in \text{SGN}[(\mathcal{D}^\top \otimes I_n)x]$ such that $\xi^\top (\mathcal{D}^\top \otimes I_n) \tilde{f} < 0$ for all $\tilde{f} \in K[f]$, as a result $\dot{\tilde{V}}_{S2}(t)$ moves into the negative half-plane of the real axis. Arbitrarily select $\eta \in \text{SGN}[(\mathcal{D}^\top \otimes I_n)x]$. Choose $\xi_k = \text{sign}[(\mathcal{D}^\top \otimes I_n)_k \cdot x]$ if $\text{sign}[(\mathcal{D}^\top \otimes I_n)_k \cdot x] \neq 0$ and choose $\xi_k = \eta_k$ if $\text{sign}[(\mathcal{D}^\top \otimes I_n)_k \cdot x] = 0$, where $\xi_k$ and $\eta_k$ represent the $k$th element in vectors $\xi$ and $\eta$, respectively. Should $\dot{\tilde{V}}_{S2}(t) \neq 0$, we assume a scenario with $\tilde{a} \in \dot{\tilde{V}}_{S2}(t)$, leading to the following analysis:
\begin{align} \nonumber
        &\tilde{a} = -\beta \left( \xi^\top (\mathcal{D}^\top \otimes I_n) \left( \nabla^2 \tilde{L}(x, t) \right)^{-1} (\mathcal{D} \otimes I_n)\eta \right) \\\nonumber
        &+ \xi^\top (\mathcal{D}^\top \otimes I_n) \Psi \\\nonumber
        &\leq -\beta \left( \xi^\top (\mathcal{D}^\top \otimes I_n) \left( \nabla^2 \tilde{L}(x, t) \right)^{-1} (\mathcal{D} \otimes I_n)\xi \right)\\\nonumber
        &+ \xi^\top (\mathcal{D}^\top \otimes I_n) \Psi\\
        &\leq -\beta \bar{\lambda}_{\min} \left( (\nabla^2 \tilde{L})^{-1} \right) \|(\mathcal{D} \otimes I_n)\xi\|_2^2 + 2 \bar{\psi}n^2|\mathcal{E}|.
    \end{align}
    If any two connected nodes $(i_2, j_2) \in \mathcal{E}$ have not the same position, $x_{i_2} \neq x_{j_2}$, then $\norm{(\mathcal{D} \otimes I_n)\xi} \geq 1$. Then it can be deduced that
    \begin{equation}
        \tilde{a} \leq -\beta \bar{\lambda}_{\min} \left( (\nabla^2 \tilde{L})^{-1} \right) + 2 \bar{\psi}n^2|\mathcal{E}|.
    \end{equation}

    If $\beta \geq \frac{2 \bar{\psi} n^2 |\mathcal{E}|}{\min_{i \in \mathcal{V}} \left\{ \bar{\lambda}_{\min} \left((\nabla^2 \tilde{L}_i)^{-1} \right)\right\}} + \epsilon = \frac{2 \bar{\psi} n^2 |\mathcal{E}|}{\bar{\lambda}_{\min}\left((\nabla^2 \tilde{L})^{-1}\right)} + \epsilon$, this ensures that if there exists an edge $(i_2, j_2) \in \mathcal{E}$ such that $x_{i_2} \neq x_{j_2}$, then $\tilde{a} \leq -\epsilon$. This condition leads to  $\dot{\tilde{V}}_{S2}(t) \leq - \epsilon$ under the same circumstances. According to the Lebesgue's theory for Riemann integrability, a function on a compact interval is Riemann integrable if and only if it is bounded and the set of its discontinuous points has measure zero \cite{ko2006mathematical}. This means that despite some irregularities at certain moments, the derivative $\dot{V}_{S2}(t)$ can still be integrated over time.

    Thus, it follows that
    \begin{equation} \label{vs2-vs20}
        V_{S2}(t) - V_{S2}(0) = \int_0^t \dot{V}_{S2}(\tau) d\tau \leq -\epsilon t
    \end{equation}
    When $t>0$ and there's a connection $(i_2, j_2) \in \mathcal{E}$ in the network with different values at both nodes $(x_{i_2} \neq x_{j_2})$, then we can deduce that
    \begin{equation}
        V_{S2}(t) = \|(\mathcal{D}^\top \otimes I_n) x\|_1 = \frac{1}{2} \sum_{i=1}^{n} \sum_{j \in \mathcal{N}_i} \|x_i - x_j\|_1.
    \end{equation}

    In other words, for $V_{S2}(t)$ to be zero, every pair of connected nodes $(i, j) \in \mathcal{E}$ must have matching positions $(x_i(t)=x_j(t))$ at time $t$. As a result, $V_{S2}(t)$ will converge to zero in a finite time, and the convergence time is smaller than the initial value of $V_{S2}(0)/\epsilon$. Furthermore, as $V_{S2}(t)$ approaches zero, $\|x_i - x_j\|_1 \rightarrow 0$ for all $i \in \mathcal{V}$ and $j \in \mathcal{N}_i$. Given that our network is connected and undirected, this condition ensures that all agents will achieve a consensus in a finite time. Specifically, there will come a moment, denoted as $T_2$ such that $\|x_i(t) - \frac{1}{N} \sum_{j=1}^{N} x_j(t)\|_2 = 0$ for all $i \in \mathcal{V}$ and for all $t > T_2$.

    Therefore, for $t \geq T_2$, the states of the system  can achieve consensus, i.e., $x_i(t)=x_j(t), \forall i,j \in \mathcal{V}$ and the system (\ref{system without ui2}) is transformed to 
    \begin{equation} \label{system nabla L}
         \dot{x}_i(t)=  - \left( \nabla^2 \tilde{L}_i(x_i, t) \right)^{-1} \left( \nabla \tilde{L}_i(x_i, t) + \frac{\partial}{\partial t} \nabla \tilde{L}_i(x_i, t) \right)
     \end{equation}
     Define the following Lyapunov function candidate as
    \begin{equation}
          V_{S3}(t) = \frac{1}{2} \left( \sum_{i=1}^{N} \nabla \Tilde{L}_i(x_i, t) \right)^\top \left( \sum_{i=1}^{N} \nabla \Tilde{L}_i(x_i, t) \right).
    \end{equation}

    Taking the derivative of $V_{S3}(t)$ with respect to the system described in (\ref{system nabla L}) results in
    \begin{equation} \begin{split} \label{vs3dot}
        \scalemath{0.85}{\dot{V}_{S3}(t) = \sum_{i=1}^{N} \nabla \Tilde{L}_i(x_i, t)^\top 
        \times \left( \sum_{i=1}^{N} \nabla^2 \Tilde{L}_i(x_i, t) \dot{x}_i + \sum_{i=1}^{N} \frac{\partial}{\partial t} \nabla \Tilde{L}_i(x_i, t) \right)}
        \end{split}
    \end{equation}
    Then by substituting (\ref{system nabla L}) into (\ref{vs3dot}), we have
    \begin{equation} \begin{split}
        \scalemath{0.85}{\dot{V}_{S3}(t)= -\left(\sum_{i=1}^{N} \nabla \Tilde{L}_i(x_i, t) \right)^\top \left(\sum_{i=1}^{N} \nabla \Tilde{L}_i(x_i, t)\right)
        =-2 V_{S3}  \leq 0.}
          \end{split}
    \end{equation}
    which indicates that $V_{S3}(t) = e^{-2t}V_{S3}(0)$ for all $t \geq 0$. It can be concluded that $V_{S3}(t)$ exponentially converges to zero, and thus, $\sum_{i=1}^{N} \nabla \tilde{L}_i(x_i, t)$ exponentially converges to $0$.  
\end{proof}
\section{numerical simulation results}

The simulations in this section demonstrate the efficiency of the theoretical insights developed in Section III, focusing on a system comprising four agents, as described in (\ref{dynamic}). We explore two distinct scenarios characterized by different communication topologies illustrated in Fig. \ref{fig:1}, each represented by a unique Laplacian matrix. The objective is to demonstrate the influence of network topology on the system's ability to suppress disturbances and achieve optimal trajectories for the global cost function. 

\subsection{Connected network depicted in Fig. 1(a) } \label{simu a}
In the first scenario, agents are connected via an undirected graph outlined by 
\begin{equation*}
  \mathcal{L}_a = [2, -1, 0, -1; -1, 2, -1, 0; 0,-1, 2, -1;-1, 0, -1, 2 ].
\end{equation*} The individual cost function for each agent is defined as follows:
\begin{align*}
    f_1(x, t) &= (x - \sin(t))^2 + 5,\\  f_2(x, t) &= (x + 3\sin(t))^2 + \cos(t)\\
    f_3(x, t) &= (x - \cos(t))^2 -5, \\ f_4(x, t)&= (x - \sin(t))^2.
\end{align*}

\begin{figure}[!t]
\centering
\subfigure[Undirected connected graph]{\includegraphics[trim = 1mm 0mm 0mm 5mm, width=0.16\textwidth]{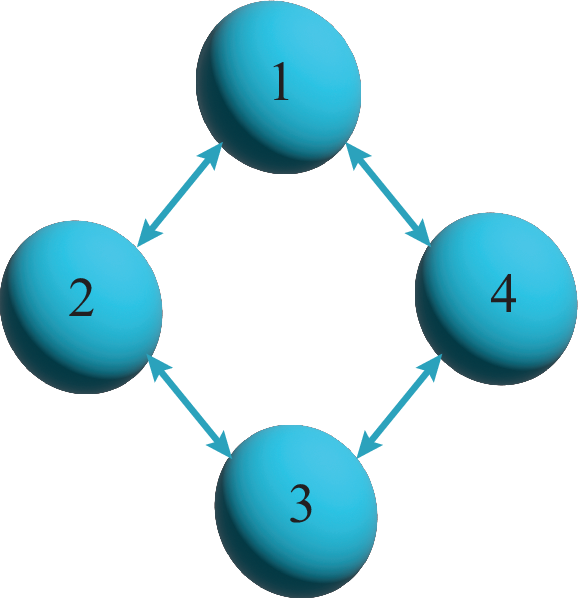}}\label{fig:1a}
\hspace{1cm}
\subfigure[Undirected less connected graph] {\includegraphics[trim = 0mm 0mm 0mm 30mm, width=0.16\textwidth]{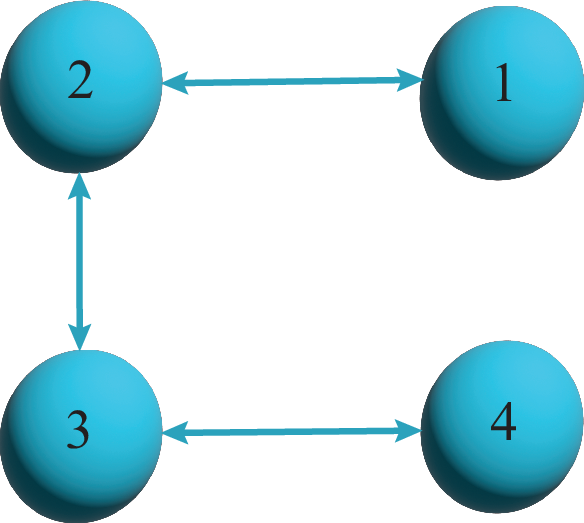}}\label{fig:1b}
\caption{Network topology of four agents} \label{fig:1}
\end{figure}

This configuration makes it evident that the optimal trajectory for the global cost function varies with time.
The system is subjected to disturbances as follows:
\begin{align*}
d_1(t) &= 3 \sin(t) + 2, \quad d_2(t) = 2 \sin(0.5\pi t), \\
d_3(t) &= 2, \hspace{1.8cm} d_4(t) = 1.5 \cos(t) + 0.5.
\end{align*}
We use the control law in (\ref{controller}) applied to system (\ref{dynamic}), with initial states set to \(x_1(0) = -2\), \(x_2(0) = -1\), \(x_3(0) = 1\), and \(x_4(0) = 3\), and control parameters as \(k_0 = 10\), \(k_1 = k_2 = \nu = 3\), \(\rho_1 = 0.5\), \(\rho_2 = \beta = 3\), and \(\rho = 18\). Consider that for agents \(j \in \mathcal{N}_i\), a constraint function is defined as \(x_j(t) - \cos(t) \leq 0\). The parameters are set to \(\rho_i(t) = 10 e^{0.05t}\), and \(\sigma_i(t) = 30 e^{-t}\) for all agents \(i \in \mathcal{V}\), thereby satisfying the initial condition (\ref{initial}). The resulting state trajectories of the agents are depicted in Fig. \ref{fig:2}(a). It illustrates effective disturbance mitigation with agents' states over time. The outcome of the constraint is depicted in Fig. \ref{fig:2}(b). During our simulation, the constraint function \(x_i(t) - \cos(t) \leq 0\), \( \forall i \in [1,\ldots,4]\). As a result, \(x_i(t) - \cos(t) - \sigma_i(t)\), \( \forall i \in [1,\ldots,4]\), consistently stayed negative.

\begin{figure*}
\centering
\subfigure[]{\includegraphics[ trim = 0mm 10mm 0mm 20mm, width=0.45\textwidth]{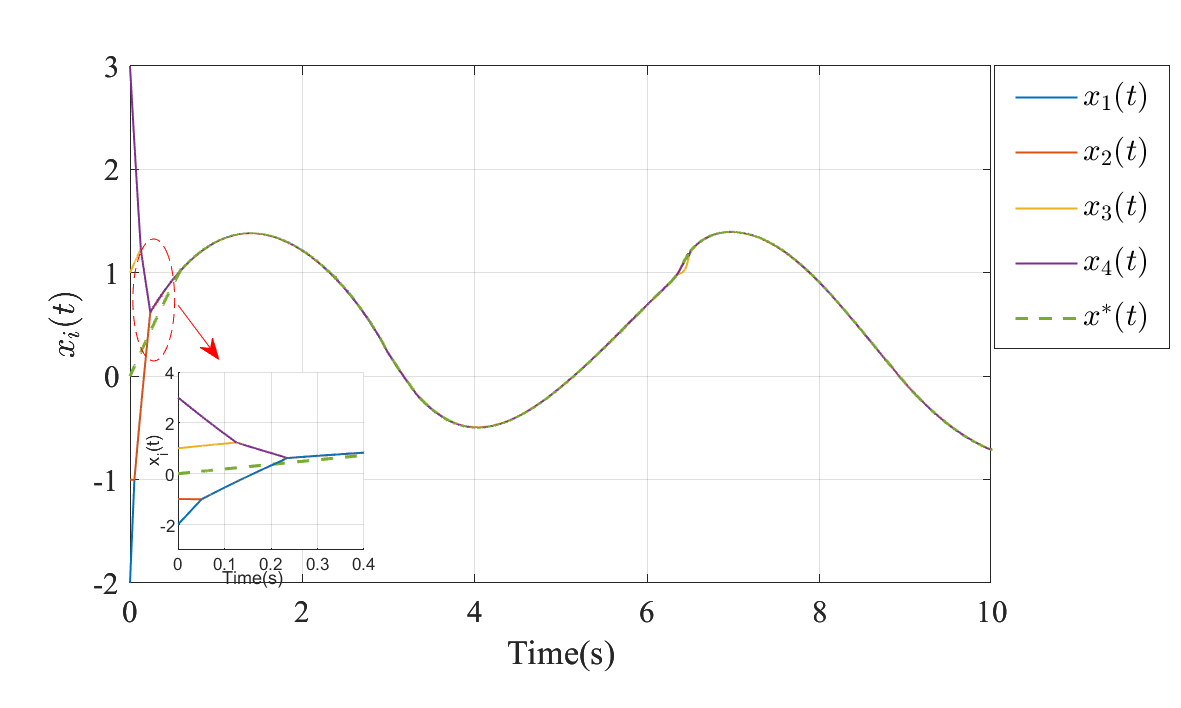}}\label{fig:2a}
\quad
\subfigure[] {\includegraphics[trim = 0mm 10mm 0mm 20mm, width=0.45\textwidth]{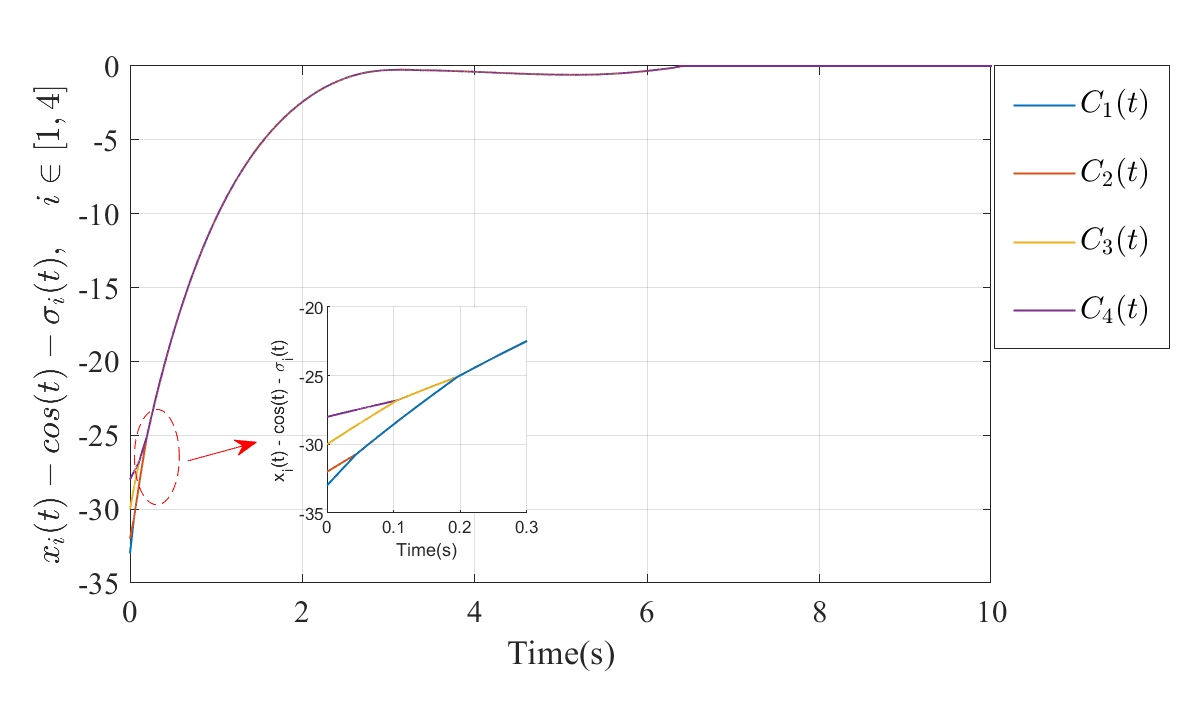}}\label{fig:2b}
\caption{Simulation results of all agents with first-order dynamics under the proposed controller. (a) State trajectories. (b) The constraint results.} \label{fig:2}
\end{figure*}

\begin{figure*}
\centering
\subfigure[]{\includegraphics[ trim = 0mm 10mm 0mm 10mm, width=0.45\textwidth]{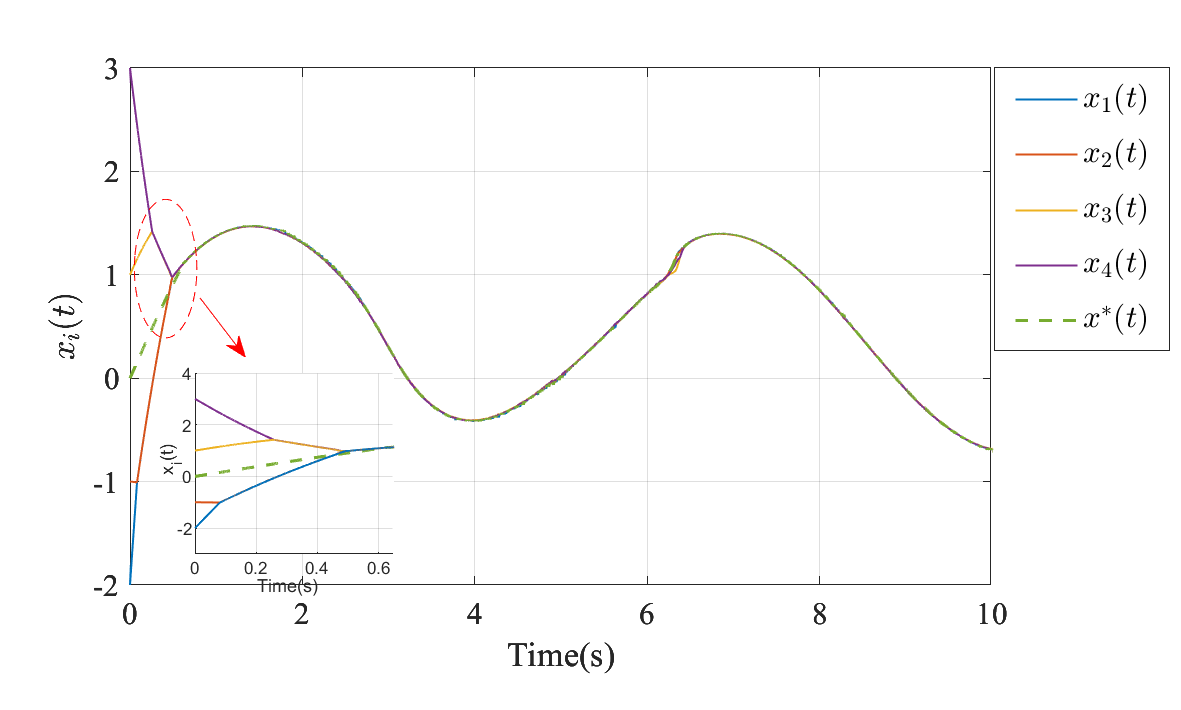}}\label{fig:3a}
\quad
\subfigure[] {\includegraphics[trim = 0mm 10mm 0mm 10mm, width=0.45\textwidth]{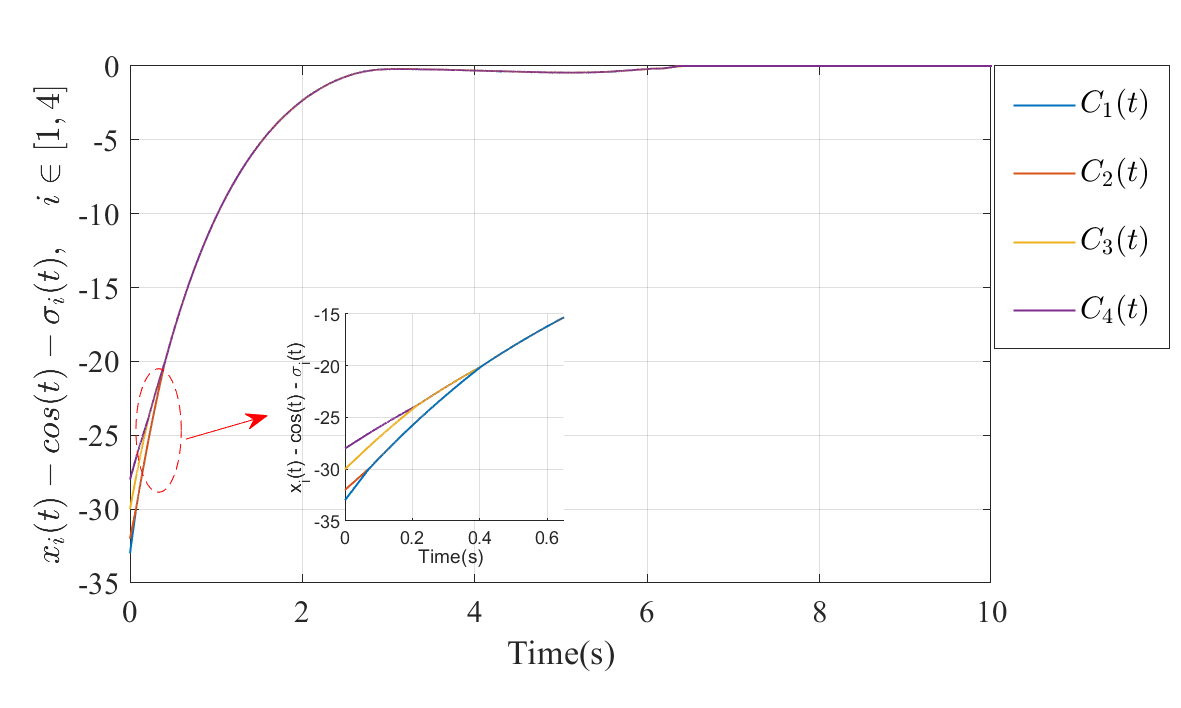}}\label{fig:3b}
\caption{Simulation results of all agents with first-order dynamics under the proposed controller. (a) State trajectories. (b) The constraint results.} \label{fig:3}
\end{figure*}

\subsection{Less-connected network depicted in Fig. 1(b)}
The second scenario features a less-connected network, described by the following Laplacian matrix:
\begin{equation*}
\mathcal{L}_b = [1, -1, 0, -0; -1, 2, -1, 0; 0,-1, 2, -1;0, 0, -1, 1 ]
\end{equation*}
The same initial conditions and parameters outlined in \ref{simu a} are considered. The results pertaining to the less connected network configuration (Fig. 1(b)) are presented in Fig \ref{fig:3}. As illustrated in Fig. \ref{fig:3}(a), all agents track the optimal trajectory which aligns with the findings of \textit{Theorem} \ref{theorem 1}. The constraint outcomes are demonstrated in Fig. \ref{fig:3}(b). However, in this less-connected network, the reduced connectivity affects the convergence speed. Specifically, there is a noticeable delay of approximately 0.25 seconds compared to the more connected scenario in \ref{simu a}. This delay can be attributed to the longer communication paths between certain agents, as fewer direct links result in slower information exchange across the network. This indicates that the degree of connectivity directly influences the system’s responsiveness.

Despite the slower convergence, the methodology remains robust in the less connected network. The agents still reach a consensus and follow the optimal trajectory, demonstrating the algorithm’s ability to perform well under reduced connectivity. This shows that the approach is adaptable to various network configurations and can tolerate less favorable connectivity without significant performance loss.

\section{conclusion}

In conclusion, this paper has developed a distributed continuous-time optimization framework to effectively handle the complexities of time-varying cost functions and constraints in multi-agent systems, particularly in the presence of disturbances. Through the integration of log-barrier functions for constraint management, coupled with an innovative integral sliding mode control for disturbance rejection, we have proposed a comprehensive solution that enhances the robustness and applicability of distributed optimization algorithms. Our works not only fill a critical gap in the current literature but also set a foundation for future exploration into more complex dynamics and situations involving disturbances with bounded derivatives.

\bibliographystyle{IEEEtran}

\begin{thebibliography}{10}
\providecommand{\url}[1]{#1}
\csname url@samestyle\endcsname
\providecommand{\newblock}{\relax}
\providecommand{\bibinfo}[2]{#2}
\providecommand{\BIBentrySTDinterwordspacing}{\spaceskip=0pt\relax}
\providecommand{\BIBentryALTinterwordstretchfactor}{4}
\providecommand{\BIBentryALTinterwordspacing}{\spaceskip=\fontdimen2\font plus
\BIBentryALTinterwordstretchfactor\fontdimen3\font minus \fontdimen4\font\relax}
\providecommand{\BIBforeignlanguage}[2]{{%
\expandafter\ifx\csname l@#1\endcsname\relax
\typeout{** WARNING: IEEEtran.bst: No hyphenation pattern has been}%
\typeout{** loaded for the language `#1'. Using the pattern for}%
\typeout{** the default language instead.}%
\else
\language=\csname l@#1\endcsname
\fi
#2}}
\providecommand{\BIBdecl}{\relax}
\BIBdecl

\bibitem{cherukuri2016initialization}
A.~Cherukuri and J.~Cortes, ``Initialization-free distributed coordination for economic dispatch under varying loads and generator commitment,'' \emph{Automatica}, vol.~74, pp. 183--193, 2016.

\bibitem{yi2022primal}
X.~Yi, S.~Zhang, T.~Yang, T.~Chai, and K.~H. Johansson, ``A primal-dual sgd algorithm for distributed nonconvex optimization,'' \emph{IEEE/CAA Journal of Automatica Sinica}, vol.~9, no.~5, pp. 812--833, 2022.

\bibitem{li2017fixed}
C.~Li, X.~Yu, X.~Zhou, and W.~Ren, ``A fixed time distributed optimization: A sliding mode perspective,'' in \emph{IECON 2017-43rd Annual Conference of the IEEE Industrial Electronics Society}.\hskip 1em plus 0.5em minus 0.4em\relax IEEE, 2017, pp. 8201--8207.

\bibitem{lin2016distributed}
P.~Lin, W.~Ren, and J.~A. Farrell, ``Distributed continuous-time optimization: nonuniform gradient gains, finite-time convergence, and convex constraint set,'' \emph{IEEE Transactions on Automatic Control}, vol.~62, no.~5, pp. 2239--2253, 2016.

\bibitem{feng2019finite}
Z.~Feng, G.~Hu, and C.~G. Cassandras, ``Finite-time distributed convex optimization for continuous-time multiagent systems with disturbance rejection,'' \emph{IEEE Transactions on Control of Network Systems}, vol.~7, no.~2, pp. 686--698, 2019.

\bibitem{huang2019distributed}
B.~Huang, Y.~Zou, Z.~Meng, and W.~Ren, ``Distributed time-varying convex optimization for a class of nonlinear multiagent systems,'' \emph{IEEE Transactions on Automatic control}, vol.~65, no.~2, pp. 801--808, 2019.

\bibitem{wangBo2020distributed}
B.~Wang, S.~Sun, and W.~Ren, ``Distributed continuous-time algorithms for optimal resource allocation with time-varying quadratic cost functions,'' \emph{IEEE Transactions on Control of Network Systems}, vol.~7, no.~4, pp. 1974--1984, 2020.

\bibitem{lee2015multirobot}
S.~G. Lee, Y.~Diaz-Mercado, and M.~Egerstedt, ``Multirobot control using time-varying density functions,'' \emph{IEEE Transactions on robotics}, vol.~31, no.~2, pp. 489--493, 2015.

\bibitem{jadbabaie2015online}
A.~Jadbabaie, A.~Rakhlin, S.~Shahrampour, and K.~Sridharan, ``Online optimization: Competing with dynamic comparators,'' in \emph{Artificial Intelligence and Statistics}.\hskip 1em plus 0.5em minus 0.4em\relax PMLR, 2015, pp. 398--406.

\bibitem{rahili2016distributed}
S.~Rahili and W.~Ren, ``Distributed continuous-time convex optimization with time-varying cost functions,'' \emph{IEEE Transactions on Automatic Control}, vol.~62, no.~4, pp. 1590--1605, 2016.

\bibitem{he2021continuous}
S.~He, X.~He, and T.~Huang, ``A continuous-time consensus algorithm using neurodynamic system for distributed time-varying optimization with inequality constraints,'' \emph{Journal of the Franklin Institute}, vol. 358, no.~13, pp. 6741--6758, 2021.

\bibitem{wang2020distributed}
Q.~Wang, Z.~Duan, Y.~Lv, Q.~Wang, and G.~Chen, ``Distributed model predictive control for linear--quadratic performance and consensus state optimization of multiagent systems,'' \emph{IEEE Transactions on Cybernetics}, vol.~51, no.~6, pp. 2905--2915, 2020.

\bibitem{sun2020distributed}
S.~Sun and W.~Ren, ``Distributed continuous-time optimization with time-varying objective functions and inequality constraints,'' in \emph{2020 59th IEEE Conference on Decision and Control (CDC)}.\hskip 1em plus 0.5em minus 0.4em\relax IEEE, 2020, pp. 5622--5627.

\bibitem{zhang2023continuous}
R.~Zhang and G.~Guo, ``Continuous distributed robust optimization of multi-agent systems with time-varying cost,'' \emph{IEEE Transactions on Control of Network Systems}, 2023.

\bibitem{sun2022distributed2023}
S.~Sun, J.~Xu, and W.~Ren, ``Distributed continuous-time algorithms for time-varying constrained convex optimization,'' \emph{IEEE Transactions on Automatic Control}, 2022.

\bibitem{edmonds2022undirected}
C.~Edmonds, ``Undirected graph theory,'' \emph{Archive of Formal Proofs}, 2022.

\bibitem{shevitz1994lyapunov}
D.~Shevitz and B.~Paden, ``Lyapunov stability theory of nonsmooth systems,'' \emph{IEEE Transactions on automatic control}, vol.~39, no.~9, pp. 1910--1914, 1994.

\bibitem{polyakov2011nonlinear}
A.~Polyakov, ``Nonlinear feedback design for fixed-time stabilization of linear control systems,'' \emph{IEEE transactions on Automatic Control}, vol.~57, no.~8, pp. 2106--2110, 2011.

\bibitem{andrieu2008homogeneous}
V.~Andrieu, L.~Praly, and A.~Astolfi, ``Homogeneous approximation, recursive observer design, and output feedback,'' \emph{SIAM Journal on control and optimization}, vol.~47, no.~4, pp. 1814--1850, 2008.

\bibitem{zuo2014new}
Z.~Zuo and L.~Tie, ``A new class of finite-time nonlinear consensus protocols for multi-agent systems,'' \emph{International Journal of Control}, vol.~87, no.~2, pp. 363--370, 2014.

\bibitem{fazlyab2017prediction}
M.~Fazlyab, S.~Paternain, V.~M. Preciado, and A.~Ribeiro, ``Prediction-correction interior-point method for time-varying convex optimization,'' \emph{IEEE Transactions on Automatic Control}, vol.~63, no.~7, pp. 1973--1986, 2017.

\bibitem{ding2023distributed}
Y.~Ding, W.~Ren, and Z.~Meng, ``Distributed continuous-time resource allocation algorithm for networked double-integrator systems with time-varying non-identical hessians and resources,'' in \emph{2023 American Control Conference (ACC)}.\hskip 1em plus 0.5em minus 0.4em\relax IEEE, 2023, pp. 1159--1164.

\bibitem{verscheure2009time}
D.~Verscheure, B.~Demeulenaere, J.~Swevers, J.~De~Schutter, and M.~Diehl, ``Time-optimal path tracking for robots: A convex optimization approach,'' \emph{IEEE Transactions on Automatic Control}, vol.~54, no.~10, pp. 2318--2327, 2009.

\bibitem{filippov2013differential}
A.~F. Filippov, \emph{Differential equations with discontinuous righthand sides: control systems}.\hskip 1em plus 0.5em minus 0.4em\relax Springer Science \& Business Media, 2013, vol.~18.

\bibitem{cortes2008discontinuous}
J.~Cortes, ``Discontinuous dynamical systems,'' \emph{IEEE Control systems magazine}, vol.~28, no.~3, pp. 36--73, 2008.

\bibitem{khalil}
H.~Khalil, ``Nonlinear systems,'' \emph{Prentice Hall}, 2002.

\bibitem{ko2006mathematical}
S.~Ko, ``Mathematical analysis,'' 2006.

\end{thebibliography}

\end{document}